\documentclass{article}

\title{When can splits be drawn in the plane?}

\author{Monika Balvo\v{c}i\={u}t\.{e}\thanks{Email: monikab@maths.otago.ac.nz}, \hspace{0.3cm}David Bryant\thanks{Email: david.bryant@otago.ac.nz; corresponding author}\\
{\small Department of Mathematics \& Statistics}\\ {\small University of Otago, P.O. Box 56, Dunedin 9054}\\ {\small New Zealand}\\~\\
Andreas Spillner\thanks{Email: mail@andreas-spillner.de}\\
{\small Department of Mathematics and Computer Science}\\ {\small University of Grieswald}}

\usepackage{mathtools}
\usepackage{amssymb}
\usepackage{amsmath}
\usepackage{amsthm}
\usepackage{mathrsfs}
\usepackage[round]{natbib}

\newtheorem*{definition}{Definition}
\newtheorem*{theorem}{Theorem}
\newtheorem*{lemma}{Lemma}

\usepackage[normalem]{ulem}
\newcommand\redout{\bgroup\markoverwith
	{\textcolor{red}{\rule[.5ex]{2pt}{0.4pt}}}\ULon}
\usepackage{xcolor}

\usepackage{graphicx}

\begin{document}

\maketitle



\begin{abstract}
	Split networks are a popular tool for the analysis and visualization of complex evolutionary histories. Every collection of splits (bipartitions) of a finite set can be represented by a split network. Here we characterize which collection of splits can be represented using a {\em planar} split network. Our main theorem links these collections of splits with oriented matroids and arrangements of lines separating points in the plane. As a consequence of our main theorem, we establish a particularly simple characterization of {\em maximal} collections of these splits. 
\end{abstract}
	

\section{Introduction}

A phylogenetic tree is an elegant and compelling model for evolutionary history, except when it isn't. Phenomena such as horizontal gene transfer, recombination, hybridization and incomplete lineage sorting all lead to evolutionary histories which cannot be represented using a single tree. Even when the underlying evolutionary history {\em is} tree-like, it can be difficult to faithfully represent uncertainty in inference using a single tree \citep{HusonBryant06}.

For these reasons, {\em phylogenetic networks} became a widely-used tool for the study of complex evolutionary histories. {\em Explicit methods}, as classified  by \cite{Huson10}, augment phylogenetic trees with extra branches and nodes to represent transfers and reticulations. {\em Implicit methods} do not attempt to reconstruct detailed reticulation histories, but instead provide an abstract portrayal of phylogenetic signals. Explicit methods aim for realism, with bells, pistons, whistles and funnels. Implicit methods aim to extract information from the data, so are analogous to signal analysis or spectral methods.

Implicit methods have proven to be the most popular in practice, at least as far as citations go. Many of the most widely used implicit methods are based on {\em split networks}. We give a formal definition of split networks below; see Fig.~\ref{fig:typesOfNetworks} for examples. See \cite{HusonBryant06} and the comprehensive monograph of  \cite{Huson10} for more introductory material on these networks. Split networks also have close links with the theory of {\em media 
graphs}, as developed independently by  \cite{Eppstein08}.

\begin{figure}
	\centering
	\includegraphics[]{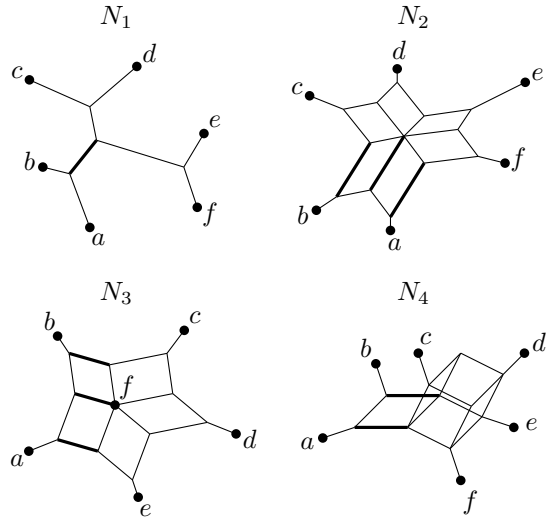}
	\caption[Examples of split networks]{Examples of split networks. 	 The class of bold edges in each correspond to the split $ab|cdef$.}
	\label{fig:typesOfNetworks}
\end{figure}

An unrooted phylogenetic tree (Fig.~\ref{fig:typesOfNetworks} $N_1$) is the simplest type of split network. Removing any edge divides the tree into two connected parts, thereby inducing a split (bipartition) on the set of leaf labels. Each edge induces a different split, and the tree can be reconstructed from these splits. 

The edges in the tree also have lengths. The tree encodes a metric on the set of leaves whereby the distance between two leaves equals the length of the path connecting them. 

In this way, phylogenetic tree represents both clustering (splits) and proximity (distance) information. However only very special collections of splits, and only very special metrics, can be represented faithfully on a tree. That is an advantage for inference if trees are appropriate representations. It is a disadvantage of they are not. The rationale behind split networks is to extend the phylogenetic trees to handle more general situations.

Networks $N_2$, $N_3$, $N_4$ in Fig.~\ref{fig:typesOfNetworks} are examples of split networks which are {\em not} unrooted phylogenetic trees. A split here corresponds to a set of edges rather than a single edge. These edges are drawn as parallel lines with the same length. 
It can be shown that every collection of splits of a finite set can be depicted using a split network \citep{Buneman71}, although the representation is not unique \citep{Wetzel95,HusonBryant06}. 

Edges in these networks have lengths. The distance between two leaves is defined as the length of the shortest path between these two leaves. The underlying graph has a property that any two shortest paths between the same two vertices will cross edges corresponding to the same set of splits. Hence the distances in the network are completely determined by the set of splits and the edge length associated with each split. This is important when it comes to setting up statistical models.

While every collection of splits can be represented by a split network, not all of these networks can be drawn usefully in the plane. The networks $N_1$, $N_2$, $N_3$ in Fig.~\ref{fig:typesOfNetworks} are all planar, but network $N_4$ is not. The principle aim of this paper is to characterize exactly when a collection of splits has a planar split network representation. Our main result connects split networks with arrangements of lines and oriented matroids. Network $N_4$ is not planar and we will be able to show that no split network representation for the same set of splits is planar. 

In the following section we give formal definitions of split networks and the other constructions, concluding with the statement of our main theorem (Theorem~\ref{thm:main} below). Section~3 gives a proof of the main theorem, while the last section examines a special maximal case which turns out to be far simpler to deal with. We finish the with some open problems.

 \section{Preliminaries}

A \emph{split} $A|B$ is a bipartition of $X$ and a \emph{split system} is a collection of splits of the same set. A split $A|B$ is \emph{proper} if both $A$ and $B$ are non-empty.

\subsection{Flat splits}\label{sec:flat}
Let $\mathcal{A}$ be a collection of lines in the plane, and let $X$ be a set of points in the plane not lying on any of these lines (Fig.~\ref{fig:flat}a). Each line $\ell \in \mathcal{A}$ partitions the plane, and therefore $X$, into at most two parts. The collection of \emph{splits} (bipartitions) of $X$ determined by the lines $\mathcal{A}$ clearly has a great deal of structure, and it is this structure which is of interest.  

\begin{figure}
	\centering
	\includegraphics{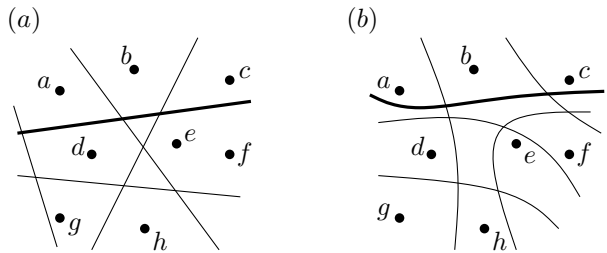}
	\caption[Flat splits]{
		A set of points $X$ and an arrangement of (a) lines and (b) pseudolines $\mathcal{A}$. None of the points in $X$ lie on any pseudoline in $\mathcal{A}$. Each $\ell \in \mathcal{A}$ divides $X$ into at most two parts inducing a split on $X$. The line in bold induces the split $abc|defgh$.}
	\label{fig:flat}
\end{figure}

Actually we consider a slightly more general situation by allowing wobbly lines. A \emph{pseudoline} in the plane is any curve homeomorphic to a line. A \emph{pseudoline arrangement} is a finite collection of pseudolines with the property that each pair of pseudolines intersect in exactly one point, and when they do intersect, they cross. A \emph{weak} arrangement of pseudolines is defined in the same way, except not every pair of pseudolines needs to intersect. Let $X$ be a set of points in the plane not lying on any pseudoline in $\mathcal{A}$ (Fig.~\ref{fig:flat}b). Each pseudoline $\ell \in \mathcal{A}$ divides $X$ into at most two parts and, as in the straight line case, induces a split of $X$.

\begin{definition}\label{defn:flatSS}\citep{BryantDress07,Spillner12} \label{def:flat}
	A split system $\mathcal{S}$ is \emph{flat} if it is induced by an arrangement of pseudolines in the plane. We say that $\mathcal{S}$ is \emph{affine} if it is induced by a collection of straight lines.
\end{definition}

Configurations of lines, pseudolines, and points arise in a wide variety of contexts, particularly in classification \citep{hastie09}, oriented matroids \citep{matroidsbook} and statistical learning theory \citep{hastie09}. Our original interest in these structures followed from applications in evolutionary biology.

\subsection{Network splits}\label{sec:net}

Let $T$ be a {\em phylogenetic tree}, a labeled tree representing the evolutionary history for a set $X$ of species or individuals and with no unlabelled vertices of degree less than three. Let $\varphi$ be the label map from  taxa in $X$ to  vertices of the tree $T$, and let $e$ be an edge of the tree. Then $T \setminus e$ consists of two connected components with corresponding vertex sets $V_1$ and $V_2$. If $A = \varphi^{-1}(V_1)$ and $B =\varphi^{-1}(V_2)$ then $A|B$ is a split that corresponds to the edge $e$ in the tree $T$. 

For many kinds of data, the signal in the set of inferred splits is more complex than can be represented by a single phylogenetic tree, motivating interest in {\em split networks}. 
%
%
%
%
%
%
%
For these networks, the underlying graph is a \emph{partial cube}, that is, an isometric subgraph of a hypercube \citep{Djokovic73}, with some vertices labeled by elements in $X$.
Partial cubes have a rich mathematical structure. It can be shown that the edge set of a partial cube can be partitioned into classes such that (a) any shortest path contains at most one edge from each class, and (b) if a shortest path between two points contains an edge in some class then so does every path between those points. Removing all edges in a single class breaks the graph into two connected components (see Fig.~\ref{fig:planar}b). Each \emph{edge class} $\varepsilon_i \in \varSigma$ induces a split $S_i \in \mathcal{S}$ of the label set $X$. Splits arising in this way are said to be \emph{induced splits} of the network $\mathcal{N}$. 


A \emph{drawing} of a split network is a straight line embedding of the graph into the plane so that edges in the same class are parallel and have the same length. A split network is \emph{planar} if it has a drawing such that edges only intersect at their endpoints, each internal cell is strictly convex, and the external face contains at least one edge from each class (Fig.~\ref{fig:planar}a). 

\begin{definition} \label{def:planar} 
	A collection of splits $\mathcal{S}$ has a \emph{planar split network representation} if there exists a planar split network, partially labelled by $X$, which induces all the splits in $\mathcal{S}$.
\end{definition}

\begin{figure}
	\centering
	\includegraphics[]{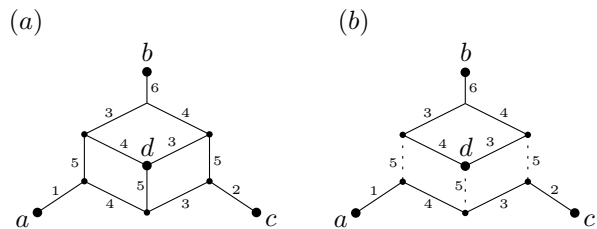}
	\caption[An example of a planar split network]{
		(a) An example of a planar split network. (b) If we remove edges of the split number five	$ac|bd$, we get two connected components: one contains vertices labeled with $\{a,c\}$ and the other with $\{b,d\}$}
	\label{fig:planar}
\end{figure}

\subsection{Oriented matroid splits}\label{sec:om}

The third type of split collection we consider arises from oriented matroid theory. An \emph{oriented matroid} is an abstract structure which mathematically can be used to represent point configurations over the reals, real hyperplane arrangements, convex polytopes and directed graphs \citep{matroidsbook}. Oriented matroids, like standard matroids, have a wide variety of different formulations and characterizations in various axiom systems. Excellent introductions to oriented matroids are found in \cite{matroidsbook} and \cite{CGbookOM}. 

Let $E$ be a finite set. A \emph{signed vector} $Y$ is a map from $E$ to $\{+,-,0\}$. Let $Y^+$, $Y^0$ and $Y^-$ denote the positive, zero, and negative indices of the sign vector $Y$. Let $Y_+$ denote the vector with $Y^+ = E$ and $Y_-$ the vector with $Y^- = E$.

For any two signed vectors $V,U$ we define their composition $V \circ U$ by
\[
V \circ U (x) = \begin{cases} V(x) & \mbox{ if $V(x) \neq 0$}\\ U(x) & \mbox{ otherwise.} \end{cases}.
\]
We say that $V$ is a \emph{restriction} of $U$ if $V(x) \neq 0$ implies $V(x) = U(x)$. The \emph{support} of $U$ is the set of elements $x \in X$ such that $U(x) \neq 0$. 

	A collection $\mathscr{T}$ of sign vectors constitute the set of topes of an oriented matroid if it satisfies the following tope axioms \citep{Handa90} :
	\begin{description}
		\item[(T0)] $\mathscr{T} \neq \emptyset$
		\item[(T1)] $T \in \mathscr{T} \Rightarrow -T \in \mathscr{T}$
		\item[(T2)] If $V$ is a restriction of some $T \in \mathscr{T}$ then either there is a $T' \in \mathscr{T}$ with $V \circ T' \in \mathscr{T}$ and $V \circ (-T') \not \in \mathscr{T}$, or $V \circ T' \in \mathscr{T}$ for every $T' \in \mathscr{T}$.
	\end{description}

Starting with the topes we can obtain other formulations of an oriented matroid. The maps $Y$ such that $Y \circ T' \in \mathscr{T}$ for every $T' \in \mathscr{T}$ are called the \emph{covectors} of the oriented matroid. The covectors with minimal support are called co-circuits. Axiom systems have been proven for sets of covectors  \citep[7.2.1]{CGbookOM} and sets of cocircuits \citep[Theorem~6.2.1]{CGbookOM}. 

The \emph{rank} $r$ of the oriented matroid is the cardinality of the smallest subset $A \subseteq X$ which intersects the support of every cocircuit. An oriented matroid is {\em uniform} if all of its cocircuits have exactly $r-1$ zero elements, i.e., $|Y^0| = r - 1$ for all $Y \in \mathcal{C}^*$ \citep{matroidsbook}. An oriented matroid is {\em acyclic} if it contains the positive tope $T_+ \in \mathscr{T}$. We will assume that the oriented matroids we work with are {\em loop-free}, meaning that every tope has support equal to the complete set of elements.

\begin{definition} \label{def:om}  A set of splits $\mathcal{S}$ is \emph{encoded by an oriented matroid of rank $3$} if there exists a set of topes $\mathscr{T}$ on $X$ of a loop-free, acyclic, rank 3 oriented matroid such that for all $A|B \in \mathcal{S}$ there is $T \in \mathscr{T}$ such that $A = T^{+}$ and $B = T^{-}$.
\end{definition}

\cite{BryantDress07} have already briefly discussed collections of splits satisfying Definition~\ref{def:om}, calling them \emph{pseudo-affine}. Later \cite{Spillner12} introduced the term \emph{flat split systems} and defined them in terms of sequences of permutations.

As an example, consider the set of signed vectors 
\begin{align*}
\mathcal{T} = \{&----, ++++, ---+, +++-, --+-, ++-+, --++,\\ &++--, +-+-, -+-+, +-++, -+--, +--+, -++- \}.
\end{align*}
Here each vector denotes a map from $a,b,c,d$ to $\{+,-\}$. It can be checked that $\mathcal{T}$ does satisfy the axioms (T0), (T1) and (T2). The corresponding collection of splits is 
\[\mathcal{S} = \{abc|d,abd|c,ab|cd,ac|bd,acd|b,ad|bc\}.\]
Note that every split corresponds to a tope and its negation, and there is no split corresponding to the positive or negative topes $T_+$ and $T_{-}$ since we assume that both sides of a split are non-empty.

\subsection{Main theorem}

Our main result is that the collections of splits just introduced are all equivalent.

\begin{theorem} \label{thm:main}
	Let $\mathcal{S}$ be a collection of splits of a finite set $X$. The following are equivalent:
	\begin{enumerate}
		\item $\mathcal{S}$ is flat, that is, $\mathcal{S}$ can be represented by an arrangement of pseudolines in the plane;
		\item $\mathcal{S}$ has a planar split network representation;
		\item $\mathcal{S}$ is encoded by a loop-free, acyclic oriented matroid of rank 3.
	\end{enumerate}
\end{theorem}

The correspondence between flat split collections and collections from oriented matroids might appear to be a straight-forward application of the celebrated Topological Representation Theorem of \cite{Folkman78}. However the representation we give is slightly different. Traditionally, the pseudolines in an arrangement correspond to the elements and the cells correspond to topes. The representation we describe has pseudolines corresponding to topes and points, lying in the cells, corresponding to elements. 

Split networks have been studied less than oriented matroids, though methods for constructing split networks have been cited thousands of times. A connection between some classes of planar split networks and line arrangements was established by \cite{Wetzel95}. He considered affine collections of splits where the set $X$ of points formed the vertices of a convex polygon. These collections are called \emph{circular}, and can be characterized by the existence of an ordering $x_1,x_2,...,x_n$ of $X$ with the property that every split in $\mathcal{S}$ has the form
\[
\{x_i,....x_{j-1}\}|X - \{x_i,...,x_{j-1}\}
\]
for some $i<j$. Later, \cite{DressHuson04} used De-Bruijn duality to prove that a collection of splits is circular if and only if it has a planar split network representation where the vertices labeled by $X$ all lie on the external face (they are \emph{planar outer labeled}, Fig.~\ref{fig:typesOfNetworks} $N_2$). Neighbor-Net \citep{BryantMoulton04} and Q-Net \citep{QNET} use this fact to produce planar split network representations of distance or quartet data.
%

There are many applications where it makes sense to construct split networks where some of the \emph{internal} vertices are also allowed to be labeled. These vertices might represent ancestral species, or spatially distributed samples. Therefore it is natural to characterize which collections of splits may be represented in this way, circular collections being a special case. 

\cite{Spillner12} made significant progress in that direction. They started with the concept of  \emph{simple allowable sequences} of permutations, as introduced by \cite{GoodmanPollack80, GoodmanPollack82}, and showed that collections of splits generated from these sequences could be represented using a planar split network. The authors stated that these split collections were equivalent to those derived from pseudoline arrangements or oriented matroids, but did not provide a proof. \cite{FLATNJ} proposed a method for computing split networks which are flat, but not necessarily circular.

\section{Proof of the main theorem}

The first step is to prove the equivalence of flat split systems and split systems from oriented matroids. 

\begin{lemma}\label{lemma:1a}
	Let $\mathcal{S}$ be a split system encoded by a loop-free, acyclic, rank 3 oriented matroid with element set $X$. Then $\mathcal{S}$ is a flat split system on $X$.
\end{lemma}

Before proving Lemma~\ref{lemma:1a}, we review two different graphical representations which exist for all rank 3 oriented matroids \citep[def. 5.3.4]{matroidsbook}:
\begin{description}
	\item[TYPE I] Arrangement of pseudolines in a projective plane.
	\item[TYPE II] Pseudoconfiguration of points.
\end{description}
See Fig.~\ref{fig:adjoint} for a type I and a type II representation of an oriented matroid $\mathcal{M}$. 
Pseudolines in the type I representation correspond to the elements in $E$ with arrows indicating orientation of each element in $\mathcal{M}$. 
The sign vector of a cell in the type I representation is determined by the orientation of the pseudolines (elements). That is, if some pseudoline is oriented towards some cell, then the corresponding sign in the covector for that cell is `$+$' and `$-$' otherwise; in case a pseudoline passes through a point that corresponds to a covector $C$, the respective sign of the covector is `0'.

Open cells of the arrangement give topes and line intersections (numbered) give cocircuits. For each tope $T \in \mathscr{T}$ at least one of $T$ and $-T$ corresponds to a cell in the type I representation. Note that only topes that correspond to open cells bounding the line at infinity $\ell_\infty$ have their negatives present in the type I representation. A similar condition is valid for the cocircuits, that is, for each cocircuit $C \in \mathcal{C}^*$ either $C$ or $-C$ is present in the type I representation.

In the type II representation of $\mathcal{M}$ the role of elements and cocircuits is reversed, i.e., cocircuits form an arrangement of pseudolines and elements correspond to the labeled points of intersection. In the type II representation each of the cocircuits is assigned an orientation. As with the type I representation, for each cocircuit either $C$ or $-C$ is present in the type II representation. Signs of each cocircuit $C$ that is present in the TYPE II representation are then determined by the relative position of each element $e_i \in E$: 
\begin{itemize}
	\item $C_i = +$ if the point corresponding to the element $e_i$ is on the positive side of the pseudoline $C$, 
	\item $C_i = -$ if the same point is on the negative side, and
	\item $C_i = 0$ if the point is on the pseudoline $C$.
\end{itemize}

\begin{figure}
	\centering
	\includegraphics{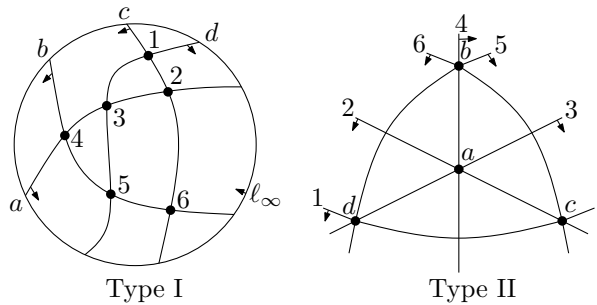}
	\caption[Type I and type II representations of an oriented matroid]{
		Type I and type II representations of an oriented matroid $\mathcal{M}$ with element set $E = \{a,b,c,d\}$ and cocircuits $\mathcal{C}^* = \{1,2,3,4,5,6\}$
	}
	\label{fig:adjoint}
\end{figure}

The existence of a type I representation is a direct result of the \cite{Folkman78} topological representation theorem reformulated for the oriented matroids of rank 3 \citep[Thm. 6.2.3]{matroidsbook}. The existence of the type II representation for the rank 3 follows from the existence of the oriented adjoint \citep[Theorem~5.3.6]{Goodman80, matroidsbook}. 
The map from the type I to the type II representation converts elements of $\mathcal{M}$ as pseudolines to elements as points and cocircuits as points to cocircuits as pseudolines. The map gives no representation for the topes of the original oriented matroid $\mathcal{M}$. Hence, to prove Lemma~\ref{lemma:1a} we augment the oriented matroid $\mathcal{M}$ to a loop-free, acyclic, rank 3 oriented matroid $\mathcal{M}_T$ such that topes of $\mathcal{M}$ are mapped to cocircuits of $\mathcal{M}_T$.

\begin{proof}[Lemma~\ref{lemma:1a}]
Let $\mathcal{S}$ be a split system encoded by a loop-free, acyclic, rank 3 oriented matroid $\mathcal{M}$ on $E = X$ with tope set $\mathscr{T}$. Construct a TYPE I representation $\mathcal{A}$ of $\mathcal{M}$, as defined in \citep[def. 5.3.4]{matroidsbook}. Each oriented pseudoline corresponds to an element and each open cell of the arrangement coresponds to a tope (Fig.~\ref{fig:adjoint}). Without loss of generality assume that one of these cells corresponds to the positive tope $T_+$. Let $\mathscr{T}' = \{T_{S}: S \in \mathcal{S}\}$ be a subset of topes $\mathscr{T}$ of $\mathcal{M}$ such that for each proper split $S = A|B$ there is exactly one tope $T_S \in \mathscr{T}'$ that induces $S$ and corresponds to a cell in $\mathcal{A}$. In case there is a choice between two cells giving topes that induce the same split $S$, we choose one at random. 

Let $P_T = \{p_S: S \in \mathcal{S}\}$ be a set of points such that  $p_{S} \in P_T$ is any point in the interior of the cell corresponding to $T_{S}$. Let $p_+$ be some point in the cell associated with $T_+$ such that $p_+ \notin P_T$. Add a set of oriented pseudolines $\{\ell_P\}$ to $\mathcal{A}$ so that at least two of them passes through each point $p \in P_T$ and all pseudolines are oriented towards $p_+$, so that the oriented matroid remains acyclic (Fig.~\ref{fig:proof1}b). Let $\mathcal{M}_T$ denote the extended oriented matroid which has this type I representation. As any intersection of two or more pseudolines corresponds to a cocircuit of $\mathcal{M}_T$ we get a bijection from $P_T$ to a subset $\mathcal{C}^*_P \subseteq \mathcal{C}^*_T$ of cocircuits of $\mathcal{M}_T$.

We now have a bijective map $\phi$ from topes in $\mathscr{T}'$ to cocircuits $\mathcal{C}^*_P$ and each tope $T \in \mathscr{T}'$ is the restriction of $\phi(T) \in \mathcal{C}^*_P$ to $X$. Take a TYPE II representation of $\mathcal{M}_T$, as defined in \citep[def. 5.3.4]{matroidsbook}. Now each cocircuit of $\mathcal{M}_T$ corresponds to a pseudoline, and each element maps to an intersection point. Each split of $\mathcal{S}$ corresponds to a tope in $\mathcal{M}$, to a point in $P_T$, to a cocircuit of $\mathcal{M}_T$ and hence to a pseudoline in the representation. Remove all pseudolines that do not correspond to topes to obtain a set of pseudolines which induces the original set of splits $\mathcal{S}$.

\end{proof}

\begin{figure}
	\centering
	\includegraphics[]{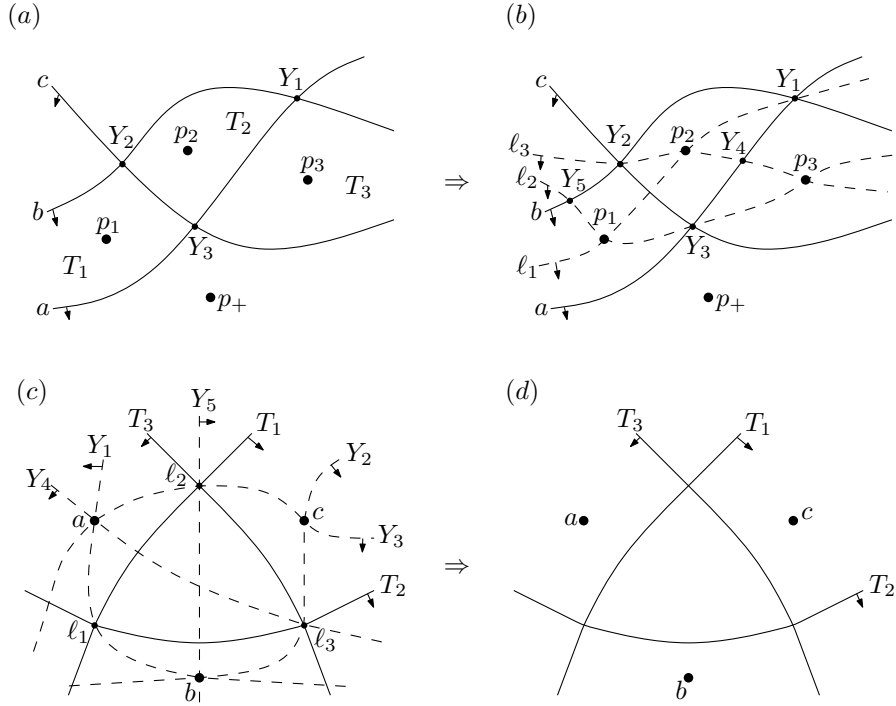}
	\caption[Correspondence between oriented matroid splits and flat splits ]{
		(a) An oriented matroid $\mathcal{M}$ on a set of elements $E = \{a,b,c\}$ with cocircuits $\mathcal{C}^* = \{Y_1, Y_2, Y_3\}$ inducing a split system $\mathcal{S} = \{a|bc, b|ac, c|ab\}$, then $\mathscr{T}' = \{T_1, T_2, T_3\}$. 
		(b) En extended oriented matroid $\mathcal{M}_T$ on $E_T = E \cup \{\ell_1, \ell_2, \ell_3\}$ with $\mathcal{C}^*_T = \mathcal{C}^* \cup \{Y4, Y5\} \cup \mathcal{C}^*_P$ where $\mathscr{T}' \rightarrow \mathcal{C}^*_P$. 
		(c) A TYPE II representation  of $\mathcal{M}_T$. 
		(d) $\mathcal{M}_T^{ad}$ restricted to $\mathcal{C}^*_P$ and with points that are labeled with elements in $E$, i.e. a flat split system induced by the splits of the oriented matroid $\mathcal{M}$}
	\label{fig:proof1}
\end{figure}

We now prove the converse of Lemma~\ref{lemma:1a}, using similar machinery.

\begin{lemma}\label{lemma:1b}
	Let $\mathcal{S}$ be a flat split system on $X$. Then $\mathcal{S}$ is encoded by a loop-free, acyclic, rank 3 oriented matroid with element set $X$.
\end{lemma}

\begin{proof}
Suppose that $\mathcal{S}$ is a flat split system, induced by a set of pseudolines $\mathcal{A}$ separating points labeled by $X$ in the plane. By embedding the arrangement in the projective plane we can repeatedly apply Levi's enlargement lemma \cite[Prop. 6.3.4]{Levi26, matroidsbook} until every point labeled by $X$ lies on at least two pseudolines. Orient these lines arbitrarily, making sure that all of them point towards one open cell and let $\mathcal{M}$ be the corresponding loop-free, acyclic, rank 3 oriented matroid. Let $\mathcal{M}^{ad}$ be an adjoint of $\mathcal{M}$, as defined in \cite{matroidsbook}.

Every $x \in X$ corresponds to a cocircuit of $\mathcal{M}$ and therefore an element $e_x$ of $\mathcal{M}^{ad}$. Furthermore, for every signed pseudoline $\ell \in \mathcal{A}$ there is a cocircuit $C_\ell^{ad}$ of $\mathcal{M}^{ad}$ such that 
\[C_{\ell}^{ad}(e_x) = \ell(x).\]
We now restrict  $\mathcal{M}^{ad}$ to elements $\{e_x:x \in X\}$ so that for each $\ell \in \mathcal{A}$ the cocircuit $C_\ell^{ad}$ restricts to a tope $T$ which induces the same split of $\{e_x:x \in X\}$ as $\ell$ does of the points labeled by $X$.

\end{proof}

We now prove the equivalence between flat splits and collections which can be represented using a planar split network. Going from flat splits to partial cubes is straight-forward: the dual of a  pseudoline arrangement is a partial cube. What is more difficult is demonstrating that this graph has a straight-line embedding in the plane where edges in the same class have the same length and are parallel. For this we apply the celebrated \emph{Bohne-Dress} theorem \citep{Bohne92}, which links zonotopal tilings and oriented matroids. Our presentation draws heavily on \cite{RichterZiegler94}.

\begin{lemma}\label{lemma:g}
	Let $\mathcal{A}$ be an arrangement of pseudolines and $X$ a set of points in the plane. Then $\mathcal{A}$ can be extended by a line $g$ such that all points in $\mathcal{A} \cup X$ are on the same side of $g$.
\end{lemma}

\begin{proof}
Suppose that $X$ is a set of points in the plane, let $\mathcal{A}$ be an arrangement of pseudolines and let $A$ be a set of points induced by $\mathcal{A}$. As a direct result of the sweeping lemma (\cite{FelsnerWeil01}, Lemma~1; \cite{Snoeyink89}, Theorem~3.1) we can add a pseudoline $g$ such that all points in $A$ are on the one side of $g$, see Fig.~\ref{fig:extending}a. What we still need to show is that $g$ can be modified so that all points in $X$ are on the same side as points in $A$. 
Let $x$ be some point in $X$ that is on the opposite side of $g$ and is closer to $g$ than any other such point. As all points in $A$ are on the same side of $g$, we get that $g$ can pass through unbounded cells of $\mathcal{A}$ only. Hence $x$ must be contained in an unbounded cell; additionally it must be one of the cells that $g$ passes through. Then we can locally perturb $g$ to go over the point $x$, see Fig.~\ref{fig:extending}b.

\begin{figure}
	\centering
	\includegraphics[]{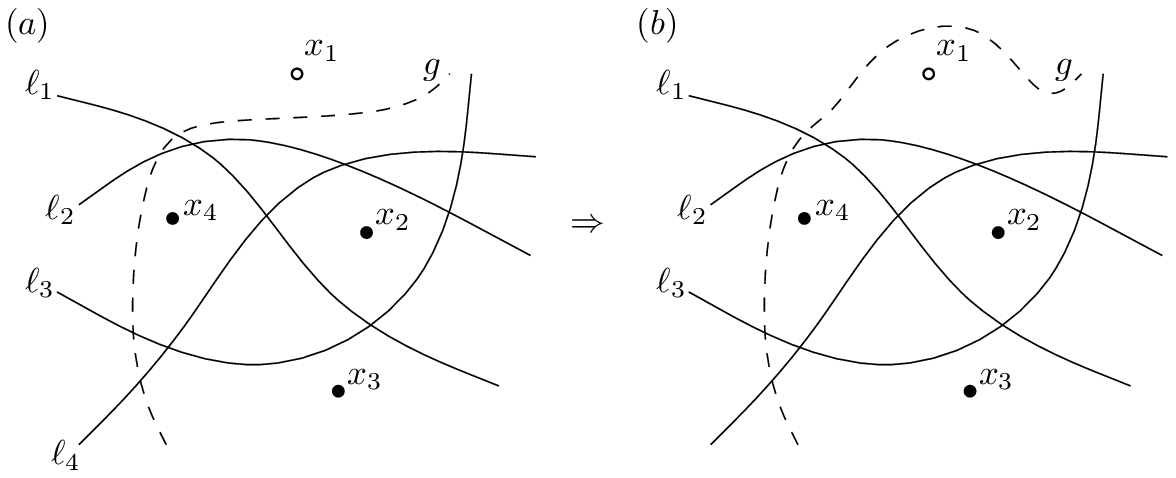}
	\caption[Extending an arrangement of pseudolines]{
		(a) Extending an arrangement of pseudolines $\mathcal{A} = \{\ell_1, \ell_2, \ell_3, \ell_4\}$ with a pseudoline $g$ such that all points in $\mathcal{A}$ are on the one side of $g$. Point $x_1$ that does not belong to $\mathcal{A}$ is on the opposite side of $g$.
		(b) Modifying $g$ so that all points in $X = \{x_1,x_2,x_3,x_4\}$ are on the same side of $g$ as points in $\mathcal{A}$
	}
	\label{fig:extending}
\end{figure}

\end{proof}

\begin{lemma}\label{lemma:2}
	Let $\mathcal{S}$ be a split system on $X$. If $\mathcal{S}$ is flat then $\mathcal{S}$ can be represented using a planar split network.
\end{lemma}

\begin{proof}
Suppose that $X$ is a set of points in the plane and let $\mathcal{A}$ be an arrangement of pseudolines which induces the collection of splits $\mathcal{S}$. Let $g$ be an auxiliary pseudoline that we add to $\mathcal{A}$ using Lemma~\ref{lemma:g}. Orient $g$ towards the points in $\mathcal{A} \cup X$. Let $p_+$ be some point in the open cell of $\mathcal{A}$ where $g$ goes to infinity and orient all pseudolines in $\mathcal{A}$ towards $p_+$. We obtain a loop-free, acyclic, rank 3 oriented matroid $\mathcal{M} = \mathcal{M}(\mathcal{A} \cup g)$ with the type I representation as described above. Elements of $\mathcal{M}$ correspond to the splits in $\mathcal{S}$ with $g$ representing an improper split.
Let $\widehat{\mathcal{L}}$ be the set of covectors of the resulting oriented matroid.
\begin{figure}
	\centering
	\includegraphics[]{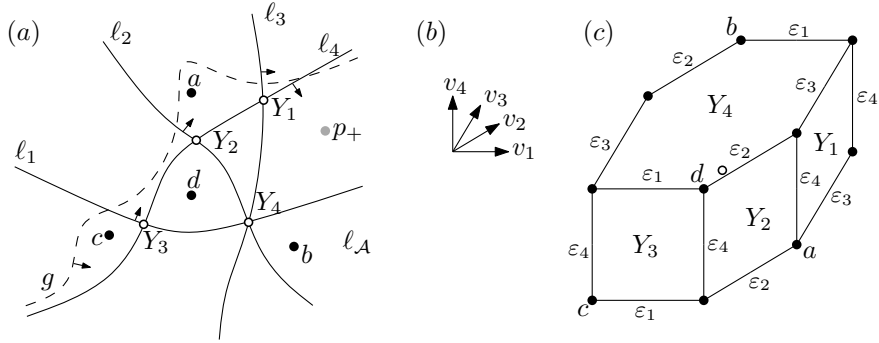}
	\caption[From flat splits to planar split network]{
		(a) An arrangement of pseudolines $\mathcal{A} = \{\ell_1, \ell_2, \ell_3, \ell_4\}$ and a set of points $\{a, b, c, d\}$. Pseudoline $g$ (dashed) is added to the $\mathcal{A}$ using Lemma~\ref{lemma:g} and oriented so that all points in $\mathcal{A}$ and $X$ are on the positive side of $g$. Pseudolines in $\mathcal{A}$ are oriented towards the dummy point $p_+$ and indexed in the order in which they intersect $g$. 
		(b) Elements in $\mathcal{A}$ are assigned vectors $\{v_1, v_2, v_3, v_4\}$ with increasing slopes. 
		(c) A zonotope of $\mathcal{A}$ as described above. White dot indicates the origin}
	\label{fig:proof2}
\end{figure}
Let $\ell_1,\ldots,\ell_n$ be an ordering of the lines in $\mathcal{A}$ given by their points of intersection with $g$, ties broken arbitrarily. Select $n$ vectors $\vec{v_1},\vec{v_2},\ldots,\vec{v_n}$ in $\Re^2$ with increasing slopes as shown in Fig.~\ref{fig:proof2}b. These constitute a realization of the contraction
\[\widehat{\mathcal{L}} / g = \Big\{Y\in \{0,+,-\}^X:(Y,0) \in \widehat{\mathcal{L}}\Big\}.  \]

A \emph{zonotope} is a polytope which is also a projection of a regular cube \citep[p. 51]{matroidsbook}. The zonotope $\mathcal{Z}(\mathcal{M})$ of an oriented matroid $\mathcal{M} = (E, \mathcal{C}^*)$ is a Minkowski sum of vectors $v \in \Re^2$ associated with elements in $E$ \citep[def. 1.1]{RichterZiegler94}:
\[
Z(\mathcal{M}) = \sum_{i = 1}^{n}[-v_i,+v_i].
\]
Zonotopes are associated with sign vectors $Y \in \mathcal{L}$ by 
\[
Z_Y = \sum_{i \in Y^0} [-v_i,+v_i] + \sum_{i \in Y^+} v_i - \sum_{i \in Y^-} v_i,
\]
see \cite{RichterZiegler94}.

Let $\mathcal{L} = \widehat{\mathcal{L}} \setminus g$ be the set of covectors after the deletion of $g$.
The split network for $\mathcal{L}$ is now constructed as follows (see Fig.~\ref{fig:proof2}c):
\begin{itemize}
	\item The vertices correspond to topes $Y$ of $\mathcal{L}$ such that $(Y,+) \in \widehat{\mathcal{L}}$, with coordinates
	\[Z_Y = \sum_{i \in Y^+} \vec{v_i} - \sum_{i \in Y^-} \vec{v_i}.\]
	Each point in $X$ is mapped to the vertex corresponding to the tope of the arrangement it is contained in.
	\item The edges correspond to covectors $Y$ of $\mathcal{L}$ such that $(Y,+) \in \widehat{\mathcal{L}}$ and $Y^0$ contains a single element. These correspond to the line segment given by the Minkowski sum 
	\[
	Z_Y = \sum_{i \in Y^+} \vec{v_i} - \sum_{i \in Y^-} \vec{v_i} + \sum_{i \in Y^0} [-\vec{v_i},\vec{v_i}].
	\]
	The zero element of these vectors $Y$ corresponds to a pseudoline of $\mathcal{A}$, and hence a split of $\mathcal{S}$. Edges corresponding to the same split induced by the pseudoline $\ell_i$ have the same direction and length as they are assigned the same vector $\vec{v_i}$.
	\item The cells correspond to cocircuits $Y$ of $\mathcal{L}$ such that $(Y,+) \in \widehat{\mathcal{L}}$ and $|Y^0| > 1$. The position of each cell is given by the same Minkowski sum
	\[
	Z_Y = \sum_{i \in Y^+} \vec{v_i} - \sum_{i \in Y^-} \vec{v_i} + \sum_{i \in Y^0} [-\vec{v_i},\vec{v_i}].
	\]
\end{itemize}

This graph is the affine projection of a partial cube \citep{FukudaHanda93}. It forms a tiling, and hence planar embedding by Theorem 4.2 of \cite{Bohne92} (see also Theorem 2.1 of \cite{RichterZiegler94}). 

\end{proof}

\cite{Eppstein05} provides a different presentation of similar ideas when proving that the region graph of an arrangement of pseudolines has a  face-symmetric planar drawing, though key steps of the proof were omitted. A relationship between marked arrangements of pseudolines and marked zonotopal tilings was also proven by \cite{FelsnerWeil01}. They proved this connection via a bijection between marked arrangements of pseudolines and allowable sequences and a bijection between allowable sequences and marked zonotopal tilings.

\begin{lemma}\label{lemma:3}
	If $\mathcal{S}$ has a planar split network representation then $\mathcal{S}$ is flat.
\end{lemma}

\begin{proof}
The first step of the proof is to show that if an internal face has edge $e$ on its boundary then it has exactly one more boundary edge in the same class. Let $C$ be the boundary of the internal face and let $e \in C$. From  Lemma 2.3 of \cite{Klavzar02}, $C$ contains at least one other edge in the same edge class as $e$, and since all the edges in this class are parallel and non-adjacent there can be at most two on the boundary of any convex region. Hence $C$ contains zero or two edges from each edge class and any two edges from the same class will be on opposite sides of the cycle.

The second step is to use this observation to construct a collection of pseudolines from the network. 
For each edge class $\mathcal{E}_i$, construct a graph $G_i$ consisting of the midpoints $\{v_e:e \in \mathcal{E}_i\}$ with edges between midpoints which lie on the same internal face (Fig. \ref{fig:netToSplits}b). There are at most two vertices in this graph with degree less than two; these correspond to the two edges in $\mathcal{E}_i$ lying on the external face.  Furthermore $G_i$ is connected, since otherwise removing the edges in $\mathcal{E}_i$ from the network would partition it into more than two components. It follows that $G_i$ is a single path, terminating in midpoints on the external face of the network.

We construct a pseudoline $\ell_i$ by taking the path determined by $G_i$ and extending the line to infinity in both directions in such a way that there are no new intersection points within  the external face of the network. By construction, the lines that we get in this way induce exactly those splits represented by the network, see Fig. \ref{fig:netToSplits}c.

\begin{figure}
	\centering
	\includegraphics[]{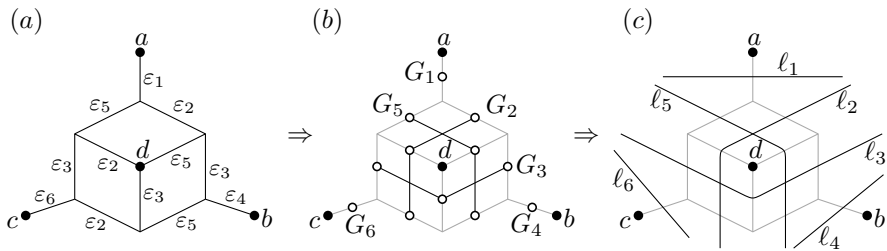}
	\caption[From planar split network to flat splits]{
		(a) A planar split network $\mathcal{N}$ on a set of taxa $\{a,b,c,d\}$ and a set of splits $\mathcal{S} = \{1,2,3,4,5,6\}$, 
		(b) a set of $G$ graphs of $\mathcal{N}$, and 
		(c) splits of $\mathcal{N}$ in the plane}
	\label{fig:netToSplits}
\end{figure}

The third step of the proof is to show that any two pseudolines in this collection intersect at most once. By the construction, the pseudolines cross every time they intersect.

Two lines $\ell_i$ and $\ell_j$ intersect if and only if there is an internal face with boundary containing edges from both $\mathcal{E}_i$ and $\mathcal{E}_j$, and when they intersect they cross. 	Suppose then that there are two internal faces $F_1,F_2$ with boundaries containing edges from both $\mathcal{E}_i$ and $\mathcal{E}_j$. Using a rotation and a shear transform we can assume that edges in $\mathcal{E}_i$ are horizontal while edges in $\mathcal{E}_j$ are vertical. 

If two internal faces intersect along a horizontal edge $e$ then one face needs to be above $e$ while the other face has to be below. It follows that $\ell_i$ is either monotonically increasing or  decreasing in the vertical coordinate. We assume that $\ell_i$ is monotonically increasing. Likewise, $\ell_j$ can be assumed to be  monotonically increasing in the horizontal coordinate.

The line $\ell_i$ partitions the plane vertically, in that it intersects every horizontal line exactly once. The pseudoline $\ell_j$ is monotonically increasing in the horizontal coordinate, so every time it crosses $\ell_i$ it crosses from left to right. It follows that $\ell_j$ crosses $\ell_i$ at most once, see Fig.~\ref{fig:intersections}. We have obtained a contradiction, and conclude that there is no more than one face with edges from both classes $\varepsilon_i$ and  $\varepsilon_j$ in a single planar split network.

\begin{figure}
	\centering
	\includegraphics{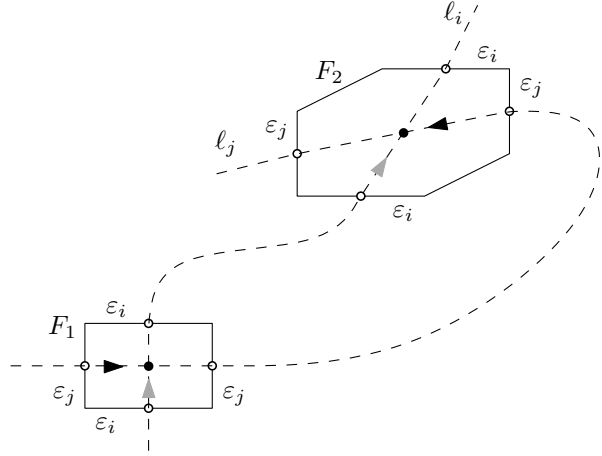}
	\caption[Sketch for the proof of Lemma~\ref{lemma:3} ]{
		Pseudoline $\ell_i$ is monotonically increasing in the vertical direction and $\ell_j$ in monotonically increasing in horizontal direction. Thus whenever $\ell_j$ crosses $\ell_i$, it crosses it from left to right ($F_1$ and $F_2$). If we want both pseudolines to cross more than once then second intersection must be from right to left what violates the horizontal monotonicity}
	\label{fig:intersections}
\end{figure}

The final step of the proof is to show that we can modify the given collection to an arrangement of pseudolines where every pair crosses \emph{exactly} once. We claim that if there is at least one pair of pseudolines in the collection which do not intersect then we can find a non-intersecting pair which can be modified to intersect in a way which affects no other pseudolines. 

Let $C$ be a simple closed  curve which contains all intersection points from the collection within its interior. 
Let $\ell_i$ and $\ell_j$ be two pseudolines which do not intersect, and let $v_i,v_i',v_j,v_j'$ be points of intersection between $\ell_i, \ell_j$ and $C$, labeled so that they appeared in the order $v_i,v_j,v_j',v_i'$ around the curve, see Fig. \ref{fig:twolines}a. 

If $v_i$ and $v_j$ are adjacent intersection points on the curve, then we can modify both $\ell_i$ and $\ell_j$ to add a point of intersection without affecting any other pseudolines in the collection, as in Fig.~\ref{fig:twolines}b. Otherwise, there is a line $\ell_k$ which intersects $C$ at some point $v_k$ between $v_i$ and $v_j$. If $\ell_k$ intersected both $\ell_i$ and $\ell_j$ then it would intersect $C$ between $v_i$ and $v_i'$ and between $v_j$ and $v_j'$, see Fig.~\ref{fig:twolines}c, a contradiction. Without loss of generality, suppose that $\ell_k$ does not intersect $\ell_i$ as shown in Fig.~\ref{fig:twolines}d. We can then repeat the argument with $\ell_i$ and $\ell_k$, noting that the number of intersection points on the curve between $v_i$ and $v_k$ is strictly less than that between $v_i$ and $v_j$. In this way we eventually obtain two non-intersecting lines with adjacent intersection points on $C$. 

\begin{figure}
	\centering
	\includegraphics[]{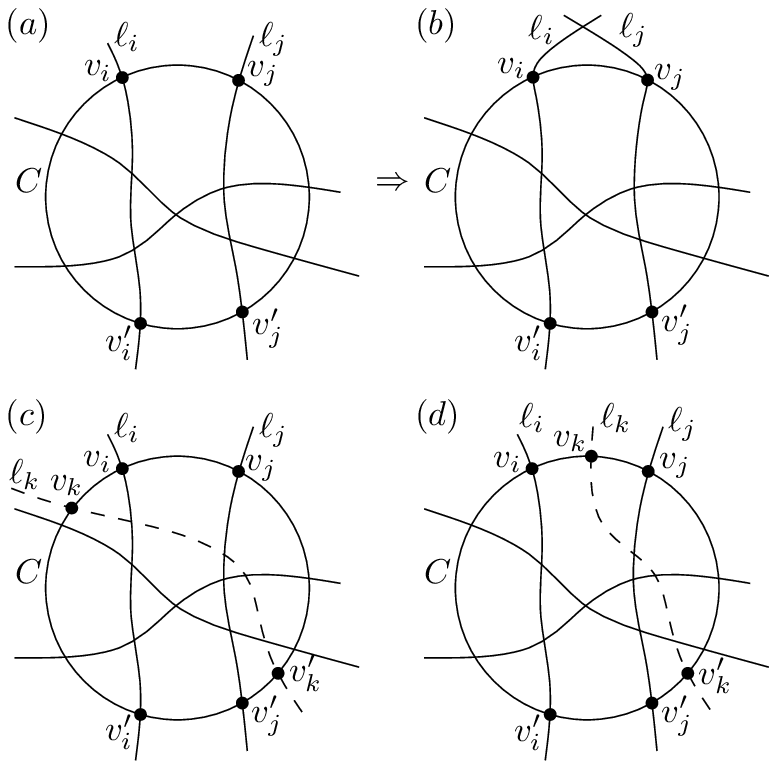}
	\caption[Arrangement of pseudolines with two non intersecting pseudolines]{
		A weak arrangement of pseudolines $\mathcal{A}$ with a closed curve $C$ bounding all intersection points. 
		(a) Two pseudolines $\ell_i$ and $\ell_j$ do not intersect. 
		(b) $\mathcal{A}$ can be modified so that $\ell_i$ and $\ell_j$ intersect without affecting any of the other pseudolines in $\mathcal{A}\setminus\{\ell_i, \ell_j\}$.
		(c) Pseudoline $\ell_k$ intersects both $\ell_i$ and $\ell_j$, thus it intersects $C$ on the different sides of $\ell_i$ and $\ell_j$. 
		(d) $\ell_k$ intersects $C$ in between $\ell_i$ and $\ell_j$, thus it cannot intersect both $\ell_i$ and $\ell_j$}
	\label{fig:twolines}
\end{figure}

\end{proof}

\section{Maximal split systems}

The equivalence between flat split systems and splits from oriented matroids allow us to easily prove properties of flat split systems which are difficult to establish directly. Here we consider properties of {\em full}  (maximal) collections of flat splits.

\begin{theorem} \label{thm:max}
Let $\mathcal{S}$ be a flat split system on a set $X$ with $n=|X|$. Then there is a flat split system $\mathcal{S}'$ on $X$ such that $\mathcal{S} \subseteq \mathcal{S}'$ and $|\mathcal{S}'| = \binom{n}{2}$.
\end{theorem}
\begin{proof}
Suppose that $|\mathcal{S}| < \binom{n}{2}$. Since $\mathcal{S}$ is flat, there is a rank 3, acyclic, loop-free oriented matroid $\mathcal{M}$ with tope set
\[\mathcal{T} = \left\{ T:T(A) = + \mbox{ and } T(B) = - \mbox{ for some $A|B \in \mathcal{S}$} \right\} \cup \{T_-,T_+\}.\]
Then 
\begin{eqnarray*}
|\mathcal{T}| &=& 2|\mathcal{S}|+2 \\
&<& 2 \binom{n-1}{2} + 2 \binom{n-1}{1} + 2.
\end{eqnarray*}
Hence $\mathcal{M}$ is not uniform \citep{Zaslavsky75, matroidsbook}. In any TYPE I representation of $\mathcal{M}$ there are pseudo-lines corresponding to three elements $a,b,c$ which all meet at a common vertex $v$. Perturbing one of these pseudo-lines gives a representation for an oriented matroid with all the topes in $\mathcal{T}$ and some additional topes. Repeating the process gives a rank 3 oriented matroid $\mathcal{M}'$ which is uniform and contains all of the topes in $\mathcal{T}$. We now let $\mathcal{S}'$ be the flat split system corresponding to $\mathcal{M}'$.

\end{proof}

We say that flat systems with $\binom{n}{2}$ are {\em full}. Theorem~\ref{thm:max} shows that every flat system can be extended to a full flat system. It turns out that these systems have a particularly simple characterization, a consequence of the equivalence with uniform matroids. Given a split system $\mathcal{S}$ on $X$ and $Y \subset X$ we define the induced split system 
\[\mathcal{S}_{|Y} = \left\{(A \cap Y) | (B \cap Y) \mbox{ such that $A \cap Y \neq \emptyset$ and $B \cap Y \neq \emptyset$ } \right\}.\]

\begin{theorem} \label{thm:fourPt}
Let $\mathcal{S}$ be a set of splits of $X$ with $|\mathcal{S}| = \binom{n}{2}$. Then $\mathcal{S}$ is flat (and hence full) if and only if for all $Y \subseteq X$ with $|Y| = 4$ the induced split system $\mathcal{S}_{|Y}$ contains exactly $6$ splits.
\end{theorem}

\begin{proof}
       Let $\mathcal{S}$ be a full flat split system on $X$. Let $\mathcal{M}$ be the corresponding acyclic, loop-free, rank 3 oriented matroid. Consider any $Y \subseteq X$ such that $|Y| = 4$, and let $\mathcal{M}_Y$ denote the oriented matroid on element set $Y$ obtained by deleting elements of $\mathcal{M}$ not in $Y$, and let $\mathcal{T}_Y$ be its set of topes. By considering a TYPE I representation of $\mathcal{M}$ we see that $\mathcal{M}_Y$ is rank $3$, acyclic, loop-free and uniform. The splits induced by $\mathcal{M}_Y$ are exactly the splits in $\mathcal{S}|_Y$.  Since $|\mathcal{T}_Y| = 14$ we have that $\mathcal{S}_{|Y}$ contains exactly six splits.

For the converse, we make use of the characterization of uniform oriented matroids in terms of VC dimension due to \cite{GartnerWelzl97}. Let $\mathcal{S}$ be a collection of splits such that $|\mathcal{S}| = \binom{n}{2}$ and $\left|\mathcal{S}_{|Y}\right| = 6$ for all four element subsets $Y \subseteq X$. Let $\mathcal{T}$ be the collection of signed vectors
\[\mathcal{T} = \left\{ T:T(A) = + \mbox{ and } T(B) = - \mbox{ for some $A|B \in \mathcal{S}$} \right\} ]\cup \{T_-,T_+\}.\]
Note that $|\mathcal{T}| = 2 \binom{n-1}{2} + 2 \binom{n-1}{1} + 2$ and $\mathcal{T} = -\mathcal{T}$. 

For any $Y$ such that $|Y| = 4$ we have that $\left| \mathcal{S}_Y \right| = 6$. Hence $\mathcal{T}$ restricted to $Y$ contains $14$ elements. In the terminology of  \citep{GartnerWelzl97}, $\mathcal{T}$ has VC dimension at most $3$, so by Theorem~50 in \citep{GartnerWelzl97} $\mathcal{T}$ is the set of topes of a uniform oriented matroid. By the definition of $\mathcal{T}$ and its cardinality, we have that this oriented matroid is rank $3$, acyclic and loop free. 

\end{proof}

A direct consequence of Theorems~\ref{thm:max} and \ref{thm:fourPt} is that if $\mathcal{S}$ is flat but not necessarily full then we will still have that $|\mathcal{S}_{|Y}| \leq 6$ for all four point subsets $Y \subseteq X$. Network $N_4$ in Fig.~\ref{fig:typesOfNetworks} depicts the splits $ab|cdef$, $abef|cd$, $abc|def$ and $acf|bde$, as well as all splits separating one element from the remainder. In this example, if $Y = \{b,c,d,f\}$ then $\mathcal{S}_{|Y}$ contains seven splits, so there exists no planar split network representing these splits. 

We note, however, that this four point condition is not a {\em sufficient} condition for a set of splits to be flat. Consider, for example, the set $\mathcal{S}$ containing splits $ab|cde$, $bc|ade$, $cd|abe$ and $ad|bce$ as well as all `trivial' splits separating one element from the remainder. The presence of these trivial splits means that any planar split network for the splits will necessarily be outer-planar, in which case if $\mathcal{S}$ was flat it would need to be circular, which it is not \citep{DressHuson04}. 

Indeed there appears to be no simple check if a collection of splits is flat or not, when it is not full. \cite{Spillner12} conjecture that the recognition of flat split systems is NP-hard, and we see no reason not to suspect this.

\subsection*{Acknowledgements}

This research was supported by an Otago doctoral scholarship to Balvo\v{c}i\={u}t\.{e}, by a Marsden grant (14-UOO-251) to Bryant, and by funds from the Allan Wilson Centre. We thank Andreas Dress and Vince Moulton for valuable discussions relating to this work.

\bibliographystyle{plainnat}

\end{document}